\documentclass[smallextended]{svjour3}     % onecolumn (second format)
\smartqed  % flush right qed marks, e.g. at end of proof
\usepackage{graphicx}
\usepackage{amsmath}
\usepackage{amssymb}
\usepackage{color}
\usepackage{subfigure}
\begin{document}
\pdfpagewidth=169.3 mm
\pdfpageheight=247.4 mm
\title{Solving the large syndrome calculation problem in steganography
\thanks{}
}
\subtitle{}


\author{Suah~Kim         \and
        Vasily~Sachnev \and
        Hyoung~Joong~Kim %etc.
}


\institute{S. Kim and H.J. Kim \at
              Graduate School of Information Security, Korea University, Seoul, South Korea \\
              \email{suahnkim@gmail.com, khj-@korea.ac.kr}           %  \\
           \and
           V. Sachnev \at School of Information, Communications, and Electronic Engineering, The Catholic University of Korea, Bucheon, South Korea
                            \email{bassvasys@hotmail.com}  
}

\date{Received: date / Accepted: date}

\maketitle

\begin{abstract}
In error correction code based image steganography, embedding using large length codes have not been researched extensively. This is due to the fact that the embedding efficiency decreases as the length becomes sufficiently larger and the memory requirement to build the parity matrix for large code is almost infeasible. However, recent studies have demonstrated that the embedding efficiency is not as important as minimizing the distortion. In light of the finding, we propose a embedding method using a large length codes which does not have such a large memory requirement. The proposed method solves the problem with the large parity matrix by embedding in the polynomial domain as oppose to matrix domain, while keeping the computational complexity equal to the matrix based methods. Furthermore, a novel embedding code called low complexity distortion minimization (LCDM) code is also presented as an example.
%
% First, we show a linear embedding scheme based on polynomial codes which eliminates the space requirement of using parity matrix, and provide an explicit formula to find all possible solutions. Second, we show the construction of LCDM code and explain the additive distortion minimization process. Unlike previous works, LCDM code's distortion minimization process is linear (memory and the computation) relative to the cover, making it efficient to embed the cover as a whole as oppose to blocks. LCDM code is designed specifically to reduce the space and embedding complexity.
\end{abstract}

\section{Introduction}
Image steganography is a data hiding technique with emphasis on the undetectability. In the past, numerous different linear error correction codes have been extensively used as an embedder in image steganography. The embedding is done like this: Linear codes provide the sender degrees of freedom to which modification can be made to the cover such that hidden message could not be be easily extracted. We refer to such modification to the cover as modifier vectors. The ideal modifier vector is the modifier vector which produces the best modified cover against steganlysis, a tool which detects whether an image has been modified or not. 

In general, an error correcting linear code can be implemented based on matrix domain or polynomial domain. Many important papers in steganography \cite{wetpaper,wetpaperimproved,largepayloads,lessdetectable,f5} discuss embedding using techniques in matrix domain as oppose to in polynomial domain. This is mostly due to the fact that the researchers in steganography are more familiar with matrix operations than the polynomial operations.

There is also no work which focus specifically on embedding using a large code, due to the fact that embedding efficiency is decreased as the code length becomes larger. However, recent studies have shown that embedding efficiency is not as important as the distortion minimization. With this in mind, the proposed paper investigates embedding using large code. 

However, syndrome calculation of a large code in matrix domain has a problem of the large parity matrix. Given that the code which is as large as the cover, the size of the parity matrix is too big for creating it.  

The paper propose a solution in which large syndrome calculation problem is by showing an efficient embedding method which can be applied in polynomial domain. In order to simplify the content, several key ideas from error correction coding is explained more in simple terms such that non-advanced reader can understand. 

For the rest of the paper, in Section 2, a literature review of the previous work is given. In Section 3, an improved implementation of embedding using polynomial code is presented. In Section 4, an example of a polynomial code called low complexity distortion minimization (LCDM) is shown. Finally, in section 5, a short example for embedding and extraction using LCDM is given. 

With regards to notations, the matrices are denoted by italicized and bolded uppercase letters, whereas polynomials are denoted by italicized uppercase letters. Each value corresponding to a coordinate in vectors or polynomials are denoted using italicized lowercase letters.
\section{Literature Review}
In this section, past embedding schemes are discussed. The first part summarizes the past works on matrix embedding schemes, whereas the second part focuses on the polynomial embedding schemes. 

One of the earliest idea of matrix embedding is suggest by Crandall \cite{crandall1998some}. Later on, Westfeld \cite{f5}'s paper showed that embedding via F5 implementation greatly reduce the amount of modification to the cover. Westfeld demonstrated the basic idea of syndrome coding; given a syndrome, there exists many modifier vectors such that all modifications to the cover using them result in the same syndrome. The ideal modifier cover is then defined to be the one with least modification. This basic idea was explored and generalized by many others like Winkler and Sch\"{o}nfeld \cite{winkler}.

Winkler and Sch\"{o}nfeld \cite{winkler} proposed a syndrome coding based on general binary linear codes $(n,k)$ for embedding. With the help from coset theory, they were able to manage to narrow down the search area for the modifying vector with a least Hamming weight from $2^{n}$ to $2^{k}$. Their method however needed all $2^{k}$ codewords; which they stored it in a look up table. Even though the improvement was significant, $2^{k}$ space requirement seemed unmanageable for large values of $k$. In addition to matrix embedding, Winkler and Sch\"{o}nfeld \cite{winkler} proposed a technique based on binary polynomial code. It used the idea of preflipping bits and finding the modifying vector with least Hamming weight exhaustively. Unlike parity matrix implementation which has a deterministic form, the polynomial implementation lack such tools and thus less efficient. They concluded that parity matrix is more efficient in terms of embedding complexity.

In realistic situations, there are elements in the cover that are not modifiable. To remedy this problem, wet paper code was developed independently. The main idea of the wet paper code was to extend current embedding techniques so that cover can be divided between locked and unlocked elements. Elements that are locked cannot be modified and unlocked elements are modifiable. The papers proposed by Fridrich et al. \cite{wetpaper,wetpaperimproved,largepayloads} explored wet paper code over non-shared selection channel. They argued that non-selection channel improved steganography security and was less vulnerable to steganalytic attacks. First, they pseudo-randomly created a matrix \textbf{\em H}, and depending on which elements are locked, submatrix \textbf{\em D} is derived by choosing corresponding columns from \textbf{\em H}. In the first paper \cite{wetpaper}, they solved a set of linear equations to find the ideal modifier vector. In the second paper \cite{wetpaperimproved}, they proposed using the meet-in-the-middle algorithm. In the third paper \cite{largepayloads}, they proposed a totally different method by finding the closest codeword from a modifier vector. They also proved that random linear codes provide good embedding efficiency and their relative embedding capacity densely covers the range of large payloads, making them ideal as an embedder \cite{largepayloads}.

In their second paper \cite{wetpaperimproved}, they explain that finding modifier vector with the least Hamming weight becomes exponentially complex for large $n$ and small $k$. Also, the space requirement of storing \textbf{\em H} is also a problem. To remedy the initial problem, they propose breaking the cover and message into smaller chunks; which effectively solved the space requirement of \textbf{\em H} as well. There are other papers \cite{rankmetriccodes} that proposed more efficient techniques to find the ideal modifier vector. 

Reed-Solomon codes are a special case of BCH (Bose Chaudhuri Hocquenghem) code. It is important to note that error correcting codes are developed so that hardware implementation is easy. This is not the case in steganography, as embedding step occurs from the software side. Fontaine and Galand \cite{reedsolomon} proposed an implementation using such codes for the wet paper coding. They showed that Reed-Solomon codes are optimal with respect to the number of locked positions. Their implementation used Lagrange interpolation and list decoding technique to optimally manage locked elements while finding all possible modifications to the cover.

Reinvestigating the problem of large look up table stated by Winkler and Sch\"{o}nfeld, Zhang et al. \cite{efficientembedder,fastbch} proposed data embedding using primitive binary BCH code. The method proposed creating two smaller look up tables to replace the $2^{k}$ table when primitive binary BCH code is used. The two smaller table is used to deterministically find solutions with modification vector with hamming weight 1 and 2. However, their method could only find up to Hamming weights of 4, and it was not clear if it it easy to extend the method for bigger Hamming weight modifier vectors.

In a different paper, Sachnev et al. \cite{lessdetectable} showed that the least Hamming weight modifier vector does not always produce better results against steganalysis. Their proposed method used a distortion function to determine the local optimum (i.e., ideal modifier vector chosen from modifier vector of weight up to 4).

Syndrome trellis code has been used as well. The algorithm proposed by Filler et al. \cite{filler2011minimizing} is based on the trellis code and it achieves better results when compared against other coding techniques. However, there is a significant tradeoff between the constraint length and the speed.

To summarize, there are two points to be made. 
\begin{itemize}

\item Many steganography techniques are based on matrix manipulation. The latest finding suggests that unlike the initial assumption, the least number of modification to the cover does not necessarily guarantee the best modified cover. The space requirement of generating large \textbf{\em H} and non-linear complexity problem of the cover embedding makes preference to block embedding.

\item Polynomial based techniques have not been explored in depth, and even the existing implementations are an exhaustive search and don't have a deterministic form to find the solutions. 
\end{itemize}

In the following section, we will propose a general polynomial embedding scheme based on polynomial codes that is efficient and can accommodate cover embedding. Then, a novel embedding code, low complexity distortion minimization (LCDM) code is proposed, which has an efficient distortion minimization process.

\section{Proposed Method: Improved Embedding Based on Polynomial Code}
When embedding using a large code, the space requirement and computational effort to finding the solution have to be considered.

Matrix based implementation methods suffer from the large parity matrix problem when using one large length code is used for embedding. For an image with $n$ pixels, using (n,k) linear code, the memory requirement for a binary parity matrix of size is $\frac{\text{matrix size} \times \text{1 bit} \times 1 \text{ Byte}}{8\text{ bit}}$ $=\frac{n\times (n-k)}{8}=\frac{n^2-n\times k}{8}$ Byte. On the other hand, polynomial based implementation only requires to store generating polynomial $G(x)$, which is only $\frac{\text{(degree of }G(x)+1)\times 1 \text{ bit} \times 1\text{ Btye} }{8 \text{ bit}}=\frac{n-k+1}{8}$ Byte. For an image size of $1000\times 700$, with embedding 0.1 bits per pixel, memory requirement of the parity matrix is 6.125 GB, whereas for generating polynomial, it is only 8.75 KB. It is clear that the difference in memory requirement for the polynomial method is much lower and feasible.

%For an image of size $n\times m$, assuming embedding 0.1 bits per pixel (bpp), the memory requirement for a binary parity matrix of size is $\frac{\text{matrix size} \times \text{1 bit} \times 1 \text{ GB}}{8\times10^9 \text{ bit}}$ $=\frac{700,000\times 70,000}{8\times10^9}=6.125$ GB. For the same example, polynomial based implementation only requires to store generating polynomial $G(x)$, which is only $\frac{\text{(degree of }G(x)+1)\times 1 \text{ bit} \times 1\text{ KB} }{8\times 10^3\text{ bit}}=\frac{70,000+1}{8\times10^3}=8.75$ KB. It is clear that the difference in memory requirement for the polynomial method is much lower and feasible.

However, unlike matrix based methods, there are no known explicit formula for for finding the modification polynomial; earlier work by Winkler and Sch\"{o}nfeld, only gave an exhaustive method, which requires $2^n$ number of trials to find the solutions.

The proposed method improves upon Winkler and Sch\"{o}nfeld's work on the polynomial based implementation by giving an explicit formula for finding the modification polynomial, which finds $2^k$ solutions.  

%\subsection{Systematic and Non-systematic Coding}

%There are two different ways to implement the polynomial codes, one is called non-systematical method \cite{Abstract}, where message does not appear as a part of the codeword, and the other is called systematical method \cite{modern}, where message appear as a part of the codeword. For example, let ``1 1 1 1" be the message. In the systematic coding, the message ``1 1 1 1" used to calculate the parity ``0 0 1", and the parity is appended to the message such as ``1 1 1 1 $\vert$ 0 0 1". However, in the non-systematic coding, the codeword from the message may look like ``0 0 0 0 1 1 0" as an example. In this case, message ``1 1 1 1" does not appear as part of the codeword.
%It should be easy to observe that for error correction, non-classical method was preferred as message appear as the codeword making decoding simple. This is a different story in steganography, as decoding of the codeword is not of importance, only the syndrome calculation is. The proposed method will use the non-systematic method.
\subsection{Setup}
Before describing the proposed method, we formally define polynomial code, in the definition of the error correction code, as a $(n,k)$ linear cyclic code, where all codewords are length of $n$ and divisible by generating polynomial $G(x)$, a polynomial with degree $n-k$. 

In the context of steganography, the syndrome, which is the remainder after dividing the cover polynomial by the $G(x)$, is the secret message of length $n-k$. 

In this paper, polynomials with coefficients from $Z_2$, denoted as $Z_2[x]$ is used to explain the concept. This idea can be easily extended to polynomials with coefficients from different Galois Fields.

Let message be $\textbf{\em M}=(m_1,\ldots,m_{n-k})\in Z_2^{n-k}$ and the original cover be $\textbf{\em J}=(j_1,\ldots,j_{n})$, where each $j_i$ takes discrete values from 0 to 255 (for 8 bit case) for an image with $n$ number of pixels. The cover is first converted into binary values using modulo 2:
 \begin{equation}
 \phi:Z_{256}^l \rightarrow Z_{2}^l \newline
 \end{equation}
 \begin{equation}
 \phi(j_1,\ldots,j_l)=(j_1\ \text{mod\ }2,\dots,j_l\ \text{mod\ }2)
 \end{equation}
The processed cover is $\textbf{\em V}=\phi(\textbf{\em J})$.
\\

A trivial bijective map from vector to polynomial representation is described as follow: $\sigma$ is a bijective map from binary vector $\textbf{\em A}\in Z_2^n$ to a polynomial over $Z_2$ with degree less than $n$ and vice versa for $\sigma^{-1}$.\\ For example, if $\textbf{\em A}=(a_1,\ldots,a_{l},0,\ldots,0)\in Z_2^n$,
 \begin{equation}
 \sigma(\textbf{\em A}) \mapsto a_{1}+\cdots+a_{l}\cdot x^{l-1}=A(x)
 \end{equation}
 \begin{equation}
 \sigma^{-1}(A(x)) \mapsto (a_1,\ldots,a_{l},0,\ldots,0)=\textbf{\em A}
 \end{equation}
where $l<n$, and `+' and `$\cdot$' represents addition and multiplication in $Z_2[x]$, respectively. 
%From now on, polynomials and vectors will be discussed interchangeably as they are equivalent objects.
\subsection{Embedding}
Embedding is done in four steps. In the first step, the base modifier polynomial $E_{base}(x)$ is evaluated. $E_{base}(x)$ is used as the basis for finding all possible modification polynomial, which would cause the syndrome to be the same as the intended message $M$. In the second step, set of all possible modifier polynomials are found. In the third step, the modifier polynomial which results in the lowest distortion is chosen. Lastly, the modification is reflected to the cover image.
\subsubsection{Base modifier polynomial}
Base modifier polynomial $E_{base}(x)$ is obtained using following equation using $E_{base}(x)$:
\begin{equation}
E_{base}(x)=rem\bigg(\frac{V(x)-M(x)}{G(x)}\bigg)
\end{equation}
where $rem$ is a function that evaluates the remainder after the long division, using operations from $Z_2[x]$. $E_{base}(x)$ is used as a base polynomial to find all other distinct solutions. 
\begin{definition}
\emph{rem}
Let $L(x), G(x)$ be polynomials from $Z_2[x]$ and suppose $L(x)=P(x)G(x)+R(x)$, where degree of $R(x)$ is less than degree of $G(x)$. Then, $rem\big(\frac{L(x)}{G(x)}\big)=R(x)$
\end{definition}
\begin{corollary}
Let $L(x), G(x)$ be polynomials from $Z_2[x]$ and degree of $L(x)$ is less than degree of $G(x)$. Then, $rem\big(\frac{L(x)}{G(x)}\big)=L(x)$
\end{corollary}
\subsubsection{All possible modifier polynomials}
Let $\mathbb{E}(x)$ be the set of all possible modifier polynomials: 
\begin{equation}
\mathbb{E}(x)=\{E_{base}(x)+F(x)G(x) \big\vert F(x)\in \mathbb{F}(x)\}
\label{findallsolution}
\end{equation} 
where $\mathbb{F}(x)$ is the set of all binary polynomials with degree less than $k$. 
 
\subsubsection{Ideal modifier polynomial}
The ideal modifier polynomial $E_{ideal}(x)\in \mathbb{E}(x)$, is the modifier polynomial which causes the least distortion. An example of this step using additive distortion function is shown in the later section as an example.
\subsubsection{Generating modified cover}
Then, the modified cover $\textbf{\em J}'$ is determined as follows: \\

1. Transform the ideal modifier polynomial into the vector form:
\begin{equation}
\sigma^{-1}(E_{ideal}(x))=\textbf{\em E}_{ideal}
\end{equation} \indent
2. Use $\textbf{\em E}_{ideal}$ and the cover $\textbf{\em J}$ to obtain the modified cover $\textbf{\em J}'$:
\begin{equation}
\textbf{\em J}'= \textbf{\em J}\oplus \textbf{\em E}_{ideal}
\end{equation}
where $\oplus$ is a bitwise XOR function. Note that the proposed method can be easily extended to accommodate modification of $\{-1,0,1\}$, as + and - are equivalent operation under $Z_2$. 

\subsection{Extraction}
When $\textbf{\em J}'$ is received, modified binary cover $V'(x)=\sigma(\phi(\textbf{\em J}))$ is recovered. The following is used to extract $M(x)$:
\begin{equation}
M(x)=rem\bigg( \frac{V'(x)}{G(x)}\bigg)
\end{equation}
Then, $\textbf{\em M}=\sigma^{-1}(M(x))$. 
\subsection{Correctness}
The fact that $\mathbb{E}(x) $ represents all possible modifier polynomial and that the embedding and extraction is correct is proven using elementary operations from $Z_2[x]$. To aid with understanding, we provide Lemma 1.

\begin{lemma}
\emph{(Remainder reduction)}
Let $P(x),G(x),L(x)$ be polynomials from $Z_2[x]$.
Then, $rem\big(\frac{G(x)P(x)+L(x)}{G(x)}\big)=rem\big(\frac{L(x)}{G(x)}\big)$
\end{lemma}
\begin{proof}
Suppose $L(x)=G(x)P_1(x)+L_1(x)$, where degree of $L_1(x)$ is less than degree of $G(x)$. Then, 
\begin{eqnarray*}
rem\bigg(\frac{G(x)P(x)+L(x)}{G(x)}\bigg)&=&rem\bigg(\frac{G(x)P(x)+G(x)P_1(x)+L_1(x)}{G(x)}\bigg)\\
&=&rem\bigg(\frac{L_1(x)}{G(x)}\bigg) \text{ [By Definition 1]}\\
&=&L_1(x) \text{ [By Corollary 1]}
\end{eqnarray*}
and $rem\big(\frac{L(x)}{G(x)}\big)=L_1(x)$\\
$\therefore$ $rem\big(\frac{G(x)P(x)+L(x)}{G(x)}\big)=rem\big(\frac{L(x)}{G(x)}\big)$ as required.
\end{proof}

\begin{theorem}
\emph{(Correctness of embedding and extraction)}
Suppose modified binary cover is $V'(x)$ and generator polynomial is $G(x)$, then message polynomial $M(x)= rem\big( \frac{V'(x)}{G(x)}\big)$
\end{theorem}
\begin{proof}
Suppose $E(x)=E_{base}(x)+F(x)G(x)$ for some $F(x)\in \mathbb{F}(x)$, and $V(x)=Q(x)G(x)+R(x)$, where degree of $R(x)$ is smaller than the degree of $G(x)$. Then, \\
\begin{eqnarray*}
rem\bigg(\frac{V'(x)}{G(x)}\bigg)&=&rem\bigg(\frac{V(x)-E(x)}{G(x)}\bigg)\\
&=&rem\bigg(\frac{Q(x)G(x)+R(x)-E_{base}(x)-F(x)G(x)}{G(x)}\bigg)\\
&=&rem\bigg(\frac{R(x)-E_{base}(x)}{G(x)}\bigg) \text{ [By Lemma 1]}\\ 
&=&rem\bigg(\frac{R(x)-rem\big(\frac{V(x)-M(x)}{G(x)}\big)}{G(x)}\bigg)\\
&=&rem\bigg(\frac{R(x)-rem\big(\frac{Q(x)G(x)+R(x)-M(x)}{G(x)}\big)}{G(x)}\bigg)\\
&=&rem\bigg(\frac{R(x)-rem\big(\frac{R(x)-M(x)}{G(x)}\big)}{G(x)}\bigg) \text{ [By Lemma 1]}\\
&=&rem\bigg(\frac{R(x)-(R(x)-M(x))}{G(x)}\bigg) \text{ [By Corollary 1]} \\
&=&rem\bigg(\frac{M(x)}{G(x)}\bigg)\\
&=&M(x) \text{ [By Corollary 1]}
\end{eqnarray*}
as required.

\begin{theorem}
 $\mathbb{E}(x)$ represent the set of all possible $2^k$ modifier polynomials.
\end{theorem}
\end{proof} 
\begin{proof}
There are $2^k$ distinct possible $E(x)\in \mathbb{E}(x)$ by definition, and all of them are valid modifier polynomials by Theorem 1. To prove that there are no other solutions, we use proof by contradiction.\\
Suppose there is a polynomial $Q(x)\notin \mathbb{E}(x)$ with degree less than $n$, but is a valid modifier polynomial. Let $V(x)=P(x)G(x)+R(x)$ and $Q(x)=P_1(x)G(x)+R_1(x)$, where degrees of $R(x)$ and $R_1(x)$ are each less then the degree of $G(x)$. Then, \\
\begin{eqnarray*}
M(x)&=&rem\bigg(\frac{V(x)-Q(x)}{G(x)}\bigg)\\
&=& rem\bigg(\frac{P(x)G(x)+R(x)-P_1(x)G(x)-R_1(x)}{G(x)}\bigg)\\
&=&R(x)-R_1(x)
\end{eqnarray*}
And,  
\begin{eqnarray*}
M(x)&=&rem\bigg(\frac{V(x)-E_{base}(x)}{G(x)}\bigg)\\
&=&R(x)-E_{base}(x)
\end{eqnarray*}

But, this would imply that $R_1(x) = E_{base}(x)$, which would mean \\ $P_1(x)G(x)+E_{base}(x)=Q(x)\in\mathbb{E}(x)$, which is a contradiction. Therefore, $E(x)\in \mathbb{E}(x)$ represent set of all valid modifier polynomials. 
(Note: Theorem 1 and 2 are quite obvious when coset theory is used, but the proofs are provided for the none advanced readers)
\end{proof}
\section{Proposed Code: Low Complexity Distortion Minimization Code}
Global distortion minimization process is computationally and memory intensive when a large code is used for embedding. In this section, a code called Low Complexity Distortion Minimization (LCDM) Code is proposed, which is specifically designed to have low complexity in distortion minimization process. 

$(n,k)$ LCDM code sets $G(x)=1+x^{n-k}$, where $n$ is the length of the cover and $n-k$ is the length of the message bits. The embedding and extraction are exactly the same as presented before, as LCDM is a polynomial code. However, LCDM code is a simple example and more research on developing a new code with restriction on $G(x)$ should be done.

For the rest of the section, the additive distortion minimization process called distortion family finding algorithm is explained and its computational complexity is discussed. 
\subsection{Distortion Minimization Process: Distortion Family Finding Algorithm}
The main idea behind LCDM is to reduce the computation in distortion minimization using the particular structure of $G(x)$ with distortion family finding algorithm (DFFA). Before discussing LCDM, two definitions are defined to help with the understanding of DFFA. 
\begin{definition}
Head polynomials of $E_{base}(x)$ are the single non-zero terms in $E_{base}(x)=x^{h_1}+\ldots+x^{h_i}+\ldots$. For example, if $E_{base}(x)= 1 + x^5$, then $1$ and $x^5$ are the two head polynomials.
\end{definition}
\begin{definition}
In $(n,k)$ LCDM code, cyclic shifts of head polynomials are polynomials which are different multiple of $n-k$ position shifts to a head polynomial $x^{h_i}$: $x^{h_i+n-k},$ $ \ldots ,x^{h_i+L(n-k)}, \ldots$, where $L \leq \frac{n-1-h_i}{n-k}$ is a number of cyclic shifts by $(n-k)$ positions. For example, let $n=15$, $n-k=4$, and $x^{h_i}=x^5$ then, $L\in\{1,2\}$ and cyclic shifts of $x^5$ are $x^{5+4},x^{5+8}$.
\end{definition}

\begin{corollary}
Cyclic shifts of $x^{h_i}$ in $(n,k)$ LCDM code is equivalent to adding $x^{h_i}$ with $\sum_{l=1}^{L}x^{h_i}G(x)x^{(l-1)(n-k)}$.
\end{corollary}
\begin{proof}
For $(n-k)$ LCDM code, $G(x)=1+x^{n-k}$, therefore \\
\begin{eqnarray}
\sum_{l=1}^{L}x^{h_i}G(x)x^{(l-1)(n-k)}&=&\sum_{l=1}^{L}x^{h_i}(1+x^{n-k})x^{(l-1)(n-k)}\\
&=&(x^{h_i}+x^{h_i+n-k}) + (x^{h_i+n-k} + x^{h_i+2(n-k)}) + \\
& &\ldots +(x^{h_i+(L-1)(n-k)} + x^{h_i+L(n-k)})\\
&=&x^{h_i}+x^{h_i+L(n-k)}
\end{eqnarray}
Therefore $x^{h_i}+\sum_{l=1}^{L}x^{h_i}G(x)x^{(l-1)(n-k)}=x^{h_i+L(n-k)}$, \textit{i.e.,} $L$ cyclic shifts of $x^{h_i}$ as required.
\end{proof}
%The Definition 2 and 3 are used to explain the process of DFFA. DFFA first divides the base modifier polynomial into its head polynomials. Then, each group contains the head polynomial and its cyclic shifted polynomials such that all polynomials in each group produce the same syndrome. Then, distortion minimization is applied within each group, and the ideal modifier polynomial is found by combining the solutions from each group.
%\begin{definition}
%\emph{Cyclic Shift}
%Cyclic shifts of $G(x)$, 
%\end{definition}
%Suppose that $D=(d_1,...,d_n)$ be the distortion vector, where $d_i \in \mathbb{R}$ represents the distortion associated with modifying the cover $j_i$ to $j_i\oplus 1$. Then, since there are $2^k$ solutions (obtained using Equation \ref{findallsolution}) which results with the syndrome equal to the message, a naive approach for distortion minimization is to compare the distortion, $D_{total}=\sum^n_{i=1}d_i\times e_i$, where $e_i$ is the $i$th coefficient of the modifier vector, for all $2^k$ modifier vectors.

%Instead, the additive property of the distortion and the structure of $G(x)$ will be exploited to break down the problem into a smaller chunks. This technique is referred as the distortion family finding algorithm.\
DFFA works like this:\\
Let $E_{base}(x)=x^{h_1}+\ldots+x^{h_i}+\ldots$.\\
For each $x^{h_i}$, do the following:
\begin{enumerate}
  \item Find the distortion associated with modifying $x^{h_i}$ and its cyclic shifts
  \item Record the position which gives the lowest distortion
\end{enumerate}
Once this is repeated for all head polynomials, let $E_{\text{ideal}}(x)$ be the polynomial with non zero terms corresponding to the positions recorded from step 2. A small example is provided in the next section to aid with understanding.
%
%First, observe that $G(x)=1+x^{n-k}$ has non-zero coefficients only at the lowest and the highest degree terms. 

Notice that polynomials within a family are all the same, because by Lemma 1, adding multiples of $G(x)$ will produce the same syndrome as the corresponding head polynomial.
\subsection{Complexity of the distortion minimization}
The computational complexity of DFFA is linear relative to $n$, length of the cover. For a each family, number of comparisons required is equal to the number of polynomials in the family \textit{i.e} $\frac{n}{n-k}$. Since this has to be done for every non-zero coefficients in $E_{base}(x)$, the total comparisons is equal to the product of number of cycles and the number of non-zero coefficients in $E_{base}(x)$. In average, $E_{base}(x)$ will have $\frac{n-k+1}{2}$ number of non-zero coefficients. Therefore, the average case complexity is $\frac{n}{n-k}\times\frac{n-k-1}{2}\doteq\frac{n}{2}$, the worst case complexity is $\frac{n}{n-k}\times(n-k-1)\doteq n$ and the best case complexity is when no modification is required. Thus distortion minimization is linear relative to $n$.
\section{Example of LCDM}
In this section, a small hands on example will be demonstrated to assist understanding.\\ \\
Let $\textbf{\em J} = ( 163, 18, 153, 20, 100, 26, 15, 212, 243, 53, 86 )$, \\
$\textbf{\em M}=(1,0 ,1)$,\\ $G(x)=1+x^3$, and \\
$\textbf{\em D}=(223, 3, 12, 4, 163, 43, 2, 12, 1, 23, 2)$\\ \\
Then, $n=11$, $n-k=3$, \\ $\textbf{\em V}=\phi(\textbf{\em J})=(1,0,1,0,0,0,1,0,1,1,0),$ \\
$V(x)=\sigma(\textbf{\em V})=1+x^2+x^6+x^8+x^9$, and\\  $M(x)=\sigma(\textbf{\em M})=1+x^2$.\\ $\therefore$ $E_{base}(x)=rem(\frac{V(x)-M(x)}{G(x)})=rem(\frac{x^6+x^8+x^9}{1+x^3})=x^2$.  \\ \\
Since $x^2$ is the only non-zero term of the base modifier polynomial, distortion family finding algorithm needs to run for only one family, \textit{i.e} $x^h=x^2$ case. \\ Cyclic shifts of $x^2$ by $n-k=3$ are: 
\begin{enumerate}
\item{$x^2$} 
\item{$x^5=x^2+x^2G(x)$}
\item{$x^8=x^2+x^2G(x)+x^2G(x)x^{3}$}
\end{enumerate}  
The distortions in respective positions are 12, 43, and 1, making $x^8$ the ideal modifier polynomial. Therefore $\textbf{\em E}_{ideal}=\sigma^{-1}(x^8)=(0,0, 0, 0, 0, 0, 0, 0, 1, 0, 0)$ and\\
 $\textbf{\em J}'=\textbf{\em J}  \oplus  E_{\text{ideal}} =(163,18,153,20,100,26,15,212,\textbf{\textcolor{red}{242}},53,86)  $.\\ \\ 
To verify that the intended $M$ can be retrieved from the modified cover $\textbf{\em J}'$:\\
$V'(x)=\sigma^{-1}(\phi(\textbf{\em J}'))=1+x^2+x^6+x^9$\\
Then, $rem(\frac{V'(x)}{G(x)})=rem(\frac{1+x^2+x^6+x^9}{1+x^3})=1+x^2$, and\\ $\sigma^{-1}(1+x^2)=(1,0,1)=\textbf{\em M}$ as required.

\section{Conclusion}
In this paper, an improved embedding technique based in polynomial domain is proposed. We proposed an improved implementation based on polynomial code, which can find all possible solutions with an explicit formula. The space requirement is dwindled down from a matrix with size $(n)\times (n-k)$ to a vector with size $n-k$. A novel embedding code called LCDM, specifically designed for steganography is also presented. The algorithm has linear complexity relative to the cover size, and therefore can be used to embed the whole cover. 

%\section{Acknowledgment}
%This work was supported in part by Institute for Information \& communications Technology Promotion (IITP) grant funded by the Korean government (MSIP) (No.2018-0-00365, Development of on-off hybrid blockchain technology for real-time large-scale data distribution), in part by under the framework of international cooperation program managed by the National Research Foundation of Korea (2018K2A9A2A06024168, FY2018).

\bibliographystyle{abbrv}
\bibliography{bibliography}  % sigproc.bib is the name of the Bibliography in this case

\begin{thebibliography}{10}

\bibitem{crandall1998some}
R.~Crandall.
\newblock Some notes on steganography.
\newblock {\em Posted on steganography mailing list}, 1998.

\bibitem{filler2011minimizing}
T.~Filler, J.~Judas, and J.~Fridrich.
\newblock Minimizing additive distortion in steganography using
  syndrome-trellis codes.
\newblock {\em Information Forensics and Security, IEEE Transactions on},
  6(3):920--935, 2011.

\bibitem{reedsolomon}
C.~Fontaine and F.~Galand.
\newblock How can reed-solomon codes improve steganographic schemes?
\newblock In {\em Proceedings of the 9th international conference on
  Information hiding}, IH'07, pages 130--144, Berlin, Heidelberg, 2007.
  Springer-Verlag.

\bibitem{wetpaper}
J.~Fridrich, M.~Goljan, P.~Lisonek, and D.~Soukal.
\newblock Writing on wet paper.
\newblock {\em Signal Processing, IEEE Transactions on}, 53(10):3923 -- 3935,
  oct. 2005.

\bibitem{wetpaperimproved}
J.~Fridrich, M.~Goljan, and D.~Soukal.
\newblock Wet paper codes with improved embedding efficiency.
\newblock {\em Information Forensics and Security, IEEE Transactions on},
  1(1):102 -- 110, march 2006.

\bibitem{largepayloads}
J.~Fridrich and D.~Soukal.
\newblock Matrix embedding for large payloads.
\newblock {\em Information Forensics and Security, IEEE Transactions on},
  1(3):390 --395, sept. 2006.

\bibitem{lessdetectable}
V.~Sachnev, H.~J. Kim, and R.~Zhang.
\newblock Less detectable jpeg steganography method based on heuristic
  optimization and bch syndrome coding.
\newblock In {\em Proceedings of the 11th ACM workshop on Multimedia and
  security}, MM\&\#38;Sec '09, pages 131--140, New York, NY, USA, 2009. ACM.

\bibitem{winkler}
D.~Sch\"{o}nfeld and A.~Winkler.
\newblock Reducing the complexity of syndrome coding for embedding.
\newblock In T.~Furon, F.~Cayre, G.~Do\"{e}rr, and P.~Bas, editors, {\em
  Information Hiding}, volume 4567 of {\em Lecture Notes in Computer Science},
  pages 145--158. Springer Berlin Heidelberg, 2007.

\bibitem{rankmetriccodes}
R.~Selvaraj and J.~Demamu.
\newblock Steganographic protocols based on rank metric codes.
\newblock In {\em Ultra Modern Telecommunications and Control Systems and
  Workshops (ICUMT), 2011 3rd International Congress on}, oct. 2011.

\bibitem{f5}
A.~Westfeld.
\newblock F5-a steganographic algorithm.
\newblock In {\em Proceedings of the 4th International Workshop on Information
  Hiding}, IHW '01, pages 289--302, London, UK, UK, 2001. Springer-Verlag.

\bibitem{efficientembedder}
R.~Zhang, V.~Sachnev, M.~Botnan, H.~J. Kim, and J.~Heo.
\newblock An efficient embedder for bch coding for steganography.
\newblock {\em Information Theory, IEEE Transactions on}, 58(12):7272 --7279,
  dec. 2012.

\bibitem{fastbch}
R.~Zhang, V.~Sachnev, and H.~Kim.
\newblock Fast bch syndrome coding for steganography.
\newblock In S.~Katzenbeisser and A.-R. Sadeghi, editors, {\em Information
  Hiding}, volume 5806 of {\em Lecture Notes in Computer Science}, pages
  48--58. Springer Berlin Heidelberg, 2009.

\end{thebibliography}
\end{document}